\documentclass[lettersize,journal]{IEEEtran}
\usepackage{amsmath,amsthm,amssymb,amsfonts}
\usepackage{algorithmic}
\usepackage{array}
\usepackage{cite}
\usepackage[caption=false,font=normalsize,labelfont=sf,textfont=sf]{subfig}
\usepackage{textcomp}
\usepackage{stfloats}
\usepackage{url}
\usepackage{verbatim}
\usepackage{graphicx}
\usepackage{cuted}
\usepackage{stfloats}
\usepackage{xcolor}
\usepackage{enumitem}
\def\BibTeX{{\rm B\kern-.05em{\sc i\kern-.025em b}\kern-.08em
		T\kern-.1667em\lower.7ex\hbox{E}\kern-.125emX}}
\usepackage{balance}
\usepackage{hyperref}
\usepackage[capitalize]{cleveref}

\newtheorem{theorem}{Theorem}

\newtheorem{lemma}{Lemma}
\begin{document}
	\title{Intra-symbol Differential Amplitude Shift Keying-aided Blind Detector for Ambient Backscatter Communication Systems}
	\author{
	Shuaijun Ma, Peng Wei, Sa Xiao, Jianquan Wang, Wanbin Tang, and
	Wei Xiang, ~\IEEEmembership{Senior Member, ~IEEE} 
	\thanks{
	This work was supported in part by the National Natural Science Foundation of China under Grant 62001094 and Grant 62301117; 
	in part by the China Postdoctoral Science Foundation under Grant 2020M683290. {\it (Corresponding author: Peng Wei.)}
	
	S. Ma, P. Wei, S. Xiao, J. Wang, and W. Tang are with the National Key Laboratory of Wireless Communications, University of Electronic Science and Technology of China, Chengdu 611731, China (e-mail: shuaijun\_ma@163.com; wppisces@uestc.edu.cn;  xiaosa@uestc.edu.cn; jqwang@uestc.edu.cn; wbtang@uestc.edu.cn).
		
	Wei Xiang is with the School of Computing, Engineering and Mathematical Sciences, La Trobe University, Melbourne, VIC 3086, Australia (e-mail: w.xiang@latrobe.edu.au).
	}}
	\maketitle
	
	\begin{abstract}
	Ambient backscatter communications (AmBC) are a promising technology for addressing the energy consumption challenge in wireless communications through the reflection or absorption of surrounding radio frequency (RF) signals. However, it grapples with the intricacies of ambient RF signal and the round-trip path loss. For traditional detectors, the incorporation of pilot sequences results in a reduction in spectral efficiency. Furthermore, traditional energy-based detectors are inherently susceptible to a notable error floor issue, attributed to the co-channel direct link interference (DLI). Consequently, this paper proposes a blind symbol detector without the prior knowledge of the channel state information, signal variance, and noise variance. By leveraging the intra-symbol differential amplitude shift keying (IDASK) scheme, this detector effectively redirects the majority of the DLI energy towards the largest eigenvalue of the received sample covariance matrix, thereby utilizing the second largest eigenvalue for efficient symbol detection. In addition, this paper conducts theoretical performance analyses of the proposed detector in terms of the false alarm probability, missed detection probability, and the bit-error rate (BER) lower bound. Simulation results demonstrate that the proposed blind detector exhibits a significant enhancement in symbol detection performance compared to its traditional counterparts.
	\end{abstract}
	
	\begin{IEEEkeywords}
	Ambient backscatter communication, Blind symbol detector, Intra-symbol differential amplitude shift keying.
	\end{IEEEkeywords}

	\section{Introduction}
	Recently, ambient backscatter communications (AmBC) have emerged as an energy-efficient technology to satisfy the low-power radio frequency (RF) requirements in green Internet of Things (IoT)\cite{IOT1,IOT2}. Unlike traditional communication system of generating RF signal for information transmission, an AmBC system uses backscatter devices (BDs) to reflect surrounding RF signals to the receiver\cite{AmBC1,AmBC2}. In addition, it can achieve high spectral utilization efficiency by sharing the same spectra with traditional communication systems. However, the signal reflected through the transmitter-BD-receiver link in the AmBC system are invariably affected by the direct link interference (DLI) from an ambient RF transmitter. Thus, successful symbol detection in the AmBC system necessitates the estimation of: (i) the direct channel from the RF transmitter to the corresponding receiver; and (ii) the cascaded channel in the transmitter-BD-receiver link. Nonetheless, the path loss of the cascaded channel is much higher than that of the direct channel \cite{hard_related}, which presenting a formidable challenge in detecting backscatter symbols.
	
	Numerous traditional detectors have been developed for AmBC systems \cite{related_cor1,related_cor2,related_cor3,related_non1,related_non2,related_non3,related_non4,related_non5,related_non6,related_semi1,related_semi2,related_semi3,related_semi4,related_semi5,related_semi6,related_semi7}. In the traditional detectors, symbols are evaluated based on the statistical properties of the received signal. According to whether the detector possesses complete knowledge of ambient RF signal and channel state information (CSI), traditional detectors can be categorized into coherent detectors and non-coherent detectors\cite{10353962}. 
	
	A coherent detector leverages prior knowledge by utilizing comprehensive ambient signal information and CSI for symbol detection. Based on this, the authors of \cite{related_cor1} proposed the maximum a posterior probability (MAP) detector and the energy-threshold determination (ETD) detector for a multi-antenna AmBC system with M-ary frequency shift keying (MFSK) modulation. These detectors effectively counter jammer attacking. Ref.\cite{related_cor3} designed the MAP detector for AmBC system using on-off keying (OOK) modulation to achieve an error-floor-free detection performance. In \cite{related_cor2}, an adaptive dual-threshold detector was proposed for frequency diverse array based AmBC systems.  
	
	Non-coherent detectors operate without requiring full knowledge of the ambient signal and complete CSI. The authors of \cite{related_semi1} proposed maximum likelihood (ML) and energy-based detectors, which utilize incomplete CSI for the AmBC system with Gaussian signals and phase shift keying (PSK) signals. Ref.\cite{related_semi2} proposed the generalized likelihood ratio test (GLRT) detector for AmBC system without the knowledge of CSI. To further enhance the detection performance in AmBC systems, the authors of \cite{related_semi3} proposed a multi-antenna AmBC signal detector based on the maximum-eigenvalue of the received signal covariance matrix. This detector employed pilot sequences to estimate the statistical variances of the received signals. Ref.\cite{related_semi4} proposed an energy-based AmBC detector for complex RF signals including complex-valued Gaussian or phase shift keying (PSK) signals, and designed the statistical variances estimator of the received signals. Ref.\cite{related_semi5} proposed an efficient detector with interleaved coding and pilot sequences, which utilizes the complex ratio to preserve the phase information. Additionally, the authors of \cite{related_semi6} employed the eigenvalue decomposition of the received signal covariance matrix for CSI estimation. Ref.\cite{related_semi7} designed a CSI estimator based on a clustering method with pilot sequences.
	
	Nonetheless, the incorporation of pilot sequences results in a reduction in spectral efficiency. To address this issue, the authors of \cite{related_non1,related_non2,related_non3} employed the differential encoders in the BD of the AmBC systems. Specifically, in \cite{related_non1}, the authors proposed an energy-aware detector and provided an analytical characterization of its achievable bit error rate (BER) performance. In \cite{related_non2}, a data-driven estimator was designed to efficiently evaluate the statistical variances of the received signals to enhance the symbol detection performance. In \cite{related_non3}, an improved detector was developed to eliminate the assumption of equal probability for symbol bits. This improvement also removes the need for an estimation process, simplifying the detection mechanism. In addition, Ref.\cite{related_non6} proposed the expectation maximization (EM) based blind CSI estimator for the AmBC system with PSK ambient signals. The EM-based signal detection method was further developed in \cite{related_non5,related_non4}. Specifically, in \cite{related_non5}, the authors proposed three detectors for the AmBC system with multi-antenna BD, and the later two are blind detectors. Additionally, they proposed optimal tag antenna selection schemes to improve the detection performance. In \cite{related_non4}, the authors proposed an AmBC system with multiple BDs, and a non-coherent parallel detection algorithm was designed to detect the symbols without requiring CSI. However, these detectors often exhibit a high error floor. 

	Motivated by the above observations, in this paper, we design a blind symbol detector for AmBC systems. The detector does not rely on the knowledge of the ambient RF signal and CSI. The main contributions of this paper can be summarized as follows
	\begin{itemize}
		\item Firstly, we propose a blind symbol detector for the AmBC system based on the second largest eigenvalue of the received signal covariance matrix. In the blind detector, the intra-symbol differential amplitude shift keying (IDASK) is employed to mitigate the DLI. Furthermore, we design a noise variance estimator based on the impact of the received signal on the second largest eigenvalue, in order to improve the estimation accuracy.
		\item Secondly, we derive the close-form expressions of the missed detection probability and false alarm probability of the proposed blind detector. Furthermore, we evaluate a lower BER bound using the total variation theory.
		\item Finally, simulation results demonstrate that the proposed detector exhibits superior detection performance in AmBC systems compared to conventional counterparts.
	\end{itemize}

	The rest of the paper is organized as follows. In Section \ref{model and coding}, we introduce the AmBC system model with a signal antenna transmitter, an IDASK-based BD, and a multi-antenna receiver. In Section \ref{detector}, the blind symbol detector based on the second largest eigenvalue is presented. In Section \ref{analyse}, the theoretical missed detection probability and false alarm probability of the proposed blind detector are analyzed, followed by the analysis of the lower bound of BER. In Section \ref{simulation}, simulation experiments are given. In Section \ref{sum}, this paper is concluded.
	
	Throughout this paper, lowercase symbols represent scalars, while boldface symbols denote vectors or matrices. $\mathcal{C}\mathcal{N}(\mu ,{{\sigma }^{2}})$ refers to the circularly symmetric complex Gaussian distribution with mean $\mu $ and variance $\sigma^2$, whereas the real Gaussian distribution with mean $\mu $ and variance $\sigma^2 $ is denoted by $\mathcal{N}(\mu ,{{\sigma }^{2}})$. ${{\mathbf{I}}_{N}}$ stands for the identity matrix of order $N$. ${{\mathbf{S}}^{H}}$, $\text{tr}(\mathbf{S})$, $\text{rank}(\mathbf{S})$ and $\text{det}\left(\mathbf{S}\right)$ represent the conjugate transpose, trace, rank and determinant of the matrix $\mathbf{S}$, respectively. $\left\| \mathbf{w} \right\|$ denotes the Euclidean norm of the vector $\mathbf{w}$. ${{\mathbf{C}}^{M\times N}}$ represents a complex matrix with $M$ rows and $N$ columns. $\exp\left(\cdot \right)$ is the exponential function. $\mathcal{Q}( \cdot)$ denotes the Q-function. $\mathbb{Z}$ stands for the set of integers. 
	\begin{figure}[t]
		\centering
		\includegraphics[width=2.5in]{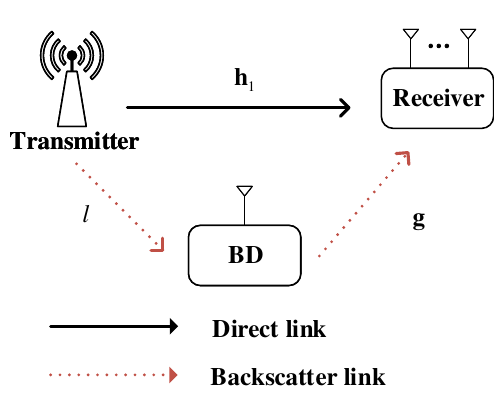}
		\caption{An AmBC system with a single-antenna RF transmitter, a single-antenna BD, and a multi-antenna receiver.}
		\label{fig1_1}
	\end{figure}
	\section{System Model} \label{model and coding}
	This paper delineates a three-node AmBC system comprising a single-antenna ambient RF transmitter, a single-antenna backscatter device (BD), and a receiver equipped with $M$ antennas. As illustrated in Fig. \ref{fig1_1}, the transmitter sends information to the receiver by modulating it onto the RF signal via an omnidirectional antenna. The BD is designed to modulate its binary symbols over the incident RF signal from the transmitter by manipulating its antenna impedance. Subsequently, the BD reflects the incident RF signal to the receiver based on the bit information. In this paper, we refer to the direct communication link between the transmitter and receiver as the direct link, and the cascade link between the transmitter (after passing BD) and the receiver as the backscatter link.
	\begin{figure*}[ht]
		\centering
		\includegraphics[width=7.2in]{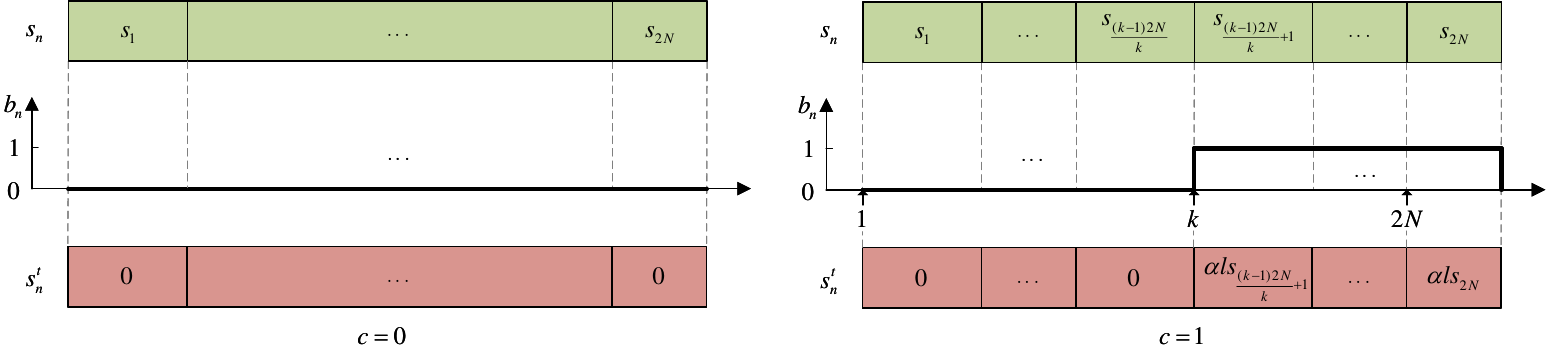}
		\caption{Schematic diagram of the IDASK scheme.}
		\label{fig1_2}
	\end{figure*}
	\subsection{IDASK-aided Signal Transmission}
	The energy of the reflected RF signal in the backscatter link is lower than that of the direct link. Thus, to capture the low-energy reflected RF signal, the BD has a lower rate than the transmitter\cite{related_semi7}. As illustrated in Fig. \ref{fig1_2}, at the top of next page, we assume that the period of one BD symbol is equal to that of $2N$ consecutive transmitter symbols. 
	
	We first assume that the channels are quasi-static block-fading according to \cite{related_semi5}. Through the transmitter-BD link, the RF signal $ \bar{s}_n $ received at the BD can be written as
	\begin{equation}
	\bar{s}_n = ls_n,
	\end{equation}
	where $n = 1,...,2N$, ${s_n}$ denotes the signal transmitted by the transmitter at time interval $n$ with variance $\sigma_s^2$, and $l$ denotes the channel coefficient of the transmitter-BD link.
	
	We then assume that the binary variables $b_n = 1 $ and $b_n = 0$ denote the reflection and non-reflection states of BD, respectively. Based on IDASK, the BD modulates the binary symbol $c \in \{0,1\}$ into the ambient RF signal by leveraging the two states of reflecting and non-reflecting its received RF signal $\bar{s}_n $. For $c = 0$, the BD remains in the non-reflection state during one BD symbol, that is, $b_n = 0$ for all $n=1, 2, \ldots, 2N$. On the contrary, for $c = 1$,
	we assume that the BD can freely switch its state between reflection and non-reflection per BD symbol. In this case, as depicted in Fig. \ref{fig1_2}, only a portion of the received RF signal in one BD symbol period is reflected. In this paper, the RF signal is not reflected in the first $(1- {1}/{k}) 2N$ transmitter symbols, that is, $b_n = 0$ for $n=1, 2, \ldots, (1- {1}/{k}) 2N $, and is reflected in the last $ {2N}/{k}$ transmitter symbols, that is, $b_n = 1$ for $n=1+(1- {1}/{k}) 2N, \ldots, 2N$. The parameter $1/k$ represents the ratio of $2N$ consecutive transmitter symbols for reflection and satisfies the condition of $2N/k \in \mathbb{Z}$ and $k \in [1,2N]$. In practice, the detection performance is independent on the order of the states of BD during one BD symbol period, it is solely contingent on the value of $k$. Thus, the BD of the IDASK-aided scheme at the BD is expressed as
	\begin{equation}
	{{b}_{n}}= \begin{cases}
	  0, & \begin{aligned}[t]
	     & n=1,...,2N, c=0, \\
	     &{\rm{or}}\ n=1,...,\left(1 - \frac{1}{k}\right)2N, c=1,
	     \end{aligned} \\
	  1, & n=\left( 1-\frac{1}{k} \right)2N+1,...,2N,c=1. \\ 
	\end{cases}
	\end{equation}
	
	Consequently, the RF signal reflected by BD can be expressed as
	\begin{equation}
	s_n^t = \alpha b_n\bar{s}_n,
	\end{equation}
	where $\alpha$ denotes the reflection coefficient of BD\cite{related_semi7}.
	\subsection{Received AmBC Signal}
	According to \cite{related_non3}, under the assumption of a short distance between the BD and receiver, the propagation delays of the direct and backscatter links are approximately equal. Consequently, the received RF signal is expressed as
	\begin{equation}
		{{\mathbf{y}}_{n}} = {{\mathbf{h}}_{1}}{{s}_{n}} + \mathbf{g} s_n^t + {{\mathbf{u}}_{n}} = \mathbf{y}_n ^b + \mathbf{y}_n ^d + {{\mathbf{u}}_{n}},
		\label{func2}
	\end{equation}
	where ${{\mathbf{h}}_{1}}\in {{\mathbf{C}}^{M\times 1}}$ and $\mathbf{g}\in {{\mathbf{C}}^{M\times 1}}$ denote the channel coefficients of the direct and BD-receiver links, respectively. We assume ${\bf{h}}_2=\alpha l{\bf{g}}$ denotes the channel coefficient of the backscatter link. Thus, we have $\mathbf{y}_n ^b = \mathbf{g}s_n ^t = b_n \mathbf{h}_2 s_n$ and $\mathbf{y}_n ^d =\mathbf{h}_1s_n $. In addition, ${{\mathbf{u}}_{n}}\sim \mathcal{C}\mathcal{N}(0,\sigma _{n}^{2}{{\mathbf{I}}_{M}})$ is the additive white Gaussian noise (AWGN) with variance $\sigma _{n}^{2}$.
	
	The hypotheses of BD sending symbols $c = 0$ and $c = 1$ are denoted by $\mathcal{H}_0$ and $\mathcal{H}_1$, respectively. The received RF signal $\mathbf{y}_n$ is expressed in matrix form as $\mathbf{Y} =[{{\mathbf{Y}}_{0}},{{\mathbf{Y}}_{1}}] \in \mathbf{C}^{M \times 2N}$, where ${{\mathbf{Y}}_{0}}=[{{\mathbf{y}}_{1}},{{\mathbf{y}}_{2}},...,{{\mathbf{y}}_{\left(1 - 1/k\right)2N}}]\in {{\mathbf{C}}^{M\times \left(1 - 1/k\right)2N}}$ and ${{\mathbf{Y}}_{1}}=[{{\mathbf{y}}_{\left(1 - 1/k\right)2N+1}},{{\mathbf{y}}_{N+2}},...,{{\mathbf{y}}_{2N}}]\in {{\mathbf{C}}^{M\times (2N/k)}}$. We assume that $s_n$ is independent and identically distributed, and ${{\mathbf{s}}_{n}}\sim \mathcal{C}\mathcal{N}(0,\sigma _{s}^{2})$. Then, under the hypothesis of ${{\mathcal{H}}_{i}}$, $\mathbf{Y}_{j} \sim \mathcal{C}\mathcal{N}(0,\mathbf{R}_{i}^{j})$, where $i,j \in \{0,1\}$, and we have
	\begin{align}
		& \mathbf{R}_{0}^{0}=\mathbb{E}[{{\mathbf{Y}}_{0}}{{\mathbf{Y}}_0^{H}}|{{\mathcal{H}}_{0}}]=\sigma _{s}^{2}{{\mathbf{h}}_1}\mathbf{h}_1^{H}+\sigma _{n}^{2}{{\mathbf{I}}_{M}} = \mathbf{R}_{1}^{0}=\mathbf{R}_{0}^{1}, \label{func3} \\
		& \mathbf{R}_{1}^{1}=\mathbb{E}[{{\mathbf{Y}}_{1}}{{\mathbf{Y}}_1^{H}}|{{\mathcal{H}}_{1}}]=\sigma _{s}^{2}{(\mathbf{h}_1 \! + \! \mathbf{h}_2)}{{(\mathbf{h}_1 \! + \! \mathbf{h}_2)}^{H}}\!+\!\sigma _{n}^{2}{{\mathbf{I}}_{M}}. \label{func4}
	\end{align}
	
	Then, the AmBC symbol detection process follows the binary hypothesis test as
	\begin{equation}
		\label{func8}
		\mathbf{Y} = \left[\mathbf{Y}_0,\mathbf{Y}_1\right] \sim \left\{ \begin{aligned}
			& \mathcal{C}\mathcal{N}(0,\mathbf{R}_{0}^{0})\cdot \mathcal{C}\mathcal{N}(0,\mathbf{R}_{0}^{1}),{{\mathcal{H}}_{0}}, \\ 
			& \mathcal{C}\mathcal{N}(0,\mathbf{R}_{1}^{0})\cdot \mathcal{C}\mathcal{N}(0,\mathbf{R}_{1}^{1}),{{\mathcal{H}}_{1}}.
		\end{aligned} \right.
	\end{equation}
	
	Based on \cite{ref23}, under the hypothesis of $\mathcal{H}_i$, the probability density function (PDF) of $\mathbf{Y}_j$ can be formulated as
	\begin{equation}\label{equ 9}
		f\left( {{\mathbf{Y}}_{j}};{{\mathcal{H}}_{i}} \right)=\frac{\exp \left\{ -\text{tr}\left( {{\left( \mathbf{R}_{i}^{j} \right)}^{-1}}{{\mathbf{Y}}_{j}}{{\mathbf{Y}}_j^{H}} \right) \right\}}{{{\pi }^{MN}}\det {{\left( \mathbf{R}_{i}^{j} \right)}^{N}}}.
	\end{equation}
	
	Based on \eqref{equ 9}, the PDF of $\mathbf{Y}$ can be derived as
	\begin{align}
		f\left( \mathbf{Y};{{\mathcal{H}}_{0}} \right) = & f\left( \mathbf{Y}_0;{{\mathcal{H}}_{0}} \right) \cdot f\left( \mathbf{Y}_1;{{\mathcal{H}}_{0}} \right) \nonumber\\
		 = & \frac{\exp \left\{ -\text{tr}\left( \left(\mathbf{R}_{0} ^0\right)^{-1}\mathbf{Y}{{\mathbf{Y}}^{H}} \right) \right\}}{{{\pi }^{2MN}}\det {{\left( {{\mathbf{R}}_{0}^0} \right)}^{2N}}}, \\
		f\left( \mathbf{Y};{{\mathcal{H}}_{1}} \right)= & \frac{\exp \left\{ -\text{tr}\left( {{\left( \mathbf{R}_{1}^{0} \right)}^{-1}}{{\mathbf{Y}}_{0}}{{\mathbf{Y}}_0^{H}}+{{\left( \mathbf{R}_{1}^{1} \right)}^{-1}}{{\mathbf{Y}}_{1}}{{\mathbf{Y}}_1^{H}} \right) \right\}}{{{\pi }^{2MN}}{{\left( \det \left( \mathbf{R}_{1}^{0} \right)\det \left( \mathbf{R}_{1}^{1} \right) \right)}^{N}}}.
	\end{align}
	
	Under the assumption of equal probabilities of symbols $c = 0$ and $c = 1$, according to \cite{related_cor3}, the general likelihood ratio test (GLRT) is utilized to detect the symbols as follows
	\begin{align}
	\label{func13}
		L\left( \mathbf{Y} \right) & =\frac{{ f\left( \mathbf{Y};{{\mathcal{H}}_{1}} \right)}}{{ f\left( \mathbf{Y};{{\mathcal{H}}_{0}} \right)}} \nonumber\\
		& = \frac{\exp \left\{ -\text{tr}\left( {{\left( \mathbf{R}_{1}^{1} \right)}^{-1}}{{\mathbf{Y}}_{1}}{{\mathbf{Y}}_1^{H}} \right) \right\}\det {{\left( \mathbf{R}_{0}^{1} \right)}^{N}}}{\exp \left\{ -\text{tr}\left( {{\left( \mathbf{R}_{0}^{1} \right)}^{-1}}{{\mathbf{Y}}_{1}}{{\mathbf{Y}}_1^{H}} \right) \right\}\det {{\left( \mathbf{R}_{1}^{1} \right)}^{N}}} \nonumber \\ & \underset{{{\mathcal{H}}_{0}}}{\overset{{{\mathcal{H}}_{1}}}{\mathop{\gtrless }}}\,1.
	\end{align}
	
	In practice, the complete information of the ambient RF signal variance, noise variance, and the CSI is unavailable. As a result, the GLRT-based detector cannot achieve the optimal symbol detection performance\cite{10353962}. Furthermore, in practical AmBC communication systems, the cascaded channel gain is considerably lower than the direct channel gain\cite{related_semi7}. Here, we define an average energy ratio $\Delta \gamma$ of the cascaded channel gain to the direct channel gain as follows
	\begin{equation}
	\Delta \gamma = \frac{\mathbb{E}\left[ {{\left\| {{\mathbf{h}}_{2}} \right\|}^{2}} \right]}{\mathbb{E}\left[ {{\left\| {{\mathbf{h}}_{1}} \right\|}^{2}} \right]}.
	\end{equation}
	As the value of $\Delta \gamma$ is close to zero, the direct link will result in significant interference to the backscatter link. To solve this problem, in the next section, we will propose a blind symbol detector based on the second largest eigenvalue of the received signal covariance matrix.
	
	\section{Second Largest Eigenvalue Based Blind Detector} \label{detector}
	In this section, we first design a blind detector based on the second largest eigenvalue of the covariance matrix of the received AmBC signal. Subsequently, a noise variance estimation algorithm is proposed.
	\begin{figure*}[hb]
	\centering
	\hrulefill 
	\vspace*{8pt}
	\begin{align}
	\label{func48}
		\det ({{\mathbf{R}}_{1}}) = &  \underbrace{\left( \sigma _{s}^{2}{{\left\| {{\mathbf{h}}_{1}} \right\|}^{2}}+\sigma _{n}^{2} \right)\left( \left( \frac{1}{k}-\frac{1}{{{k}^{2}}} \right)\sigma _{s}^{2}{{\left\| {{\mathbf{h}}_{2}} \right\|}^{2}}+\sigma _{n}^{2} \right) {{(\sigma _{n}^{2})}^{M-2}}}_{\rm{P1}} \nonumber \\ 
		& +  \underbrace{{{\left( \frac{1}{k} \right)}^{2}}\sigma _{s}^{2}\sigma _{n}^{M}{{\left\| {{\mathbf{h}}_{2}} \right\|}^{2}}+\frac{\sigma _{s}^{2}\sigma _{n}^{M}\left( \mathbf{h}_{1}^{H}{{\mathbf{h}}_{2}}+\mathbf{h}_{2}^{H}{{\mathbf{h}}_{1}} \right)}{k}-\frac{\left(k - 1\right) \sigma _{s}^{4}{{\left\| \mathbf{h}_{2}^{H}{{\mathbf{h}}_{1}} \right\|}^{2}}}{{{k}^{2}}}  {{(\sigma _{n}^{2})}^{M-2}}}_{\rm{P2}}.
		\tag{22}
	\end{align}
	\end{figure*}
	\subsection{Blind Symbol Detector}
	Based on \eqref{func3} and \eqref{func4}, the covariance matrix of the received AmBC signal can be written as
	\begin{align}\label{equ 14}
		{{\mathbf{R}}_{i}} & =\mathbb{E}[\mathbf{Y}{{\mathbf{Y}}^{H}}|{{\mathcal{H}}_{i}}] \nonumber\\
		 & = \frac{1}{k}\mathbb{E}[\mathbf{Y}_1{\mathbf{Y}_1^{H}}|{{\mathcal{H}}_{i}}] + \left(1 -\frac{1}{k}\right) \mathbb{E}[\mathbf{Y}_0{\mathbf{Y}_0^{H}}|{{\mathcal{H}}_{i}}].
	\end{align}
	
	For hypotheses of $\mathcal{H}_1$ and $\mathcal{H}_0$, Eq. \eqref{equ 14} is respectively expanded as
	\begin{align}
		{{\mathbf{R}}_{1}} = & \left(1 \! - \! \frac{1}{k} \right)\sigma _{s}^{2}{{\mathbf{h}}_1}{{\mathbf{h}} _1 ^{H}} \! + \! \frac{1}{k}\sigma _{s}^{2}{({\mathbf{h}} _ 1 \! + \! {\mathbf{h}} _ 2)}{({\mathbf{h}} _ 1 \! + \! {\mathbf{h}} _ 2)^{H}} \! + \! \sigma _{n}^{2}{{\mathbf{I}}_{M}}, \label{define_cov}\\
		{{\mathbf{R}}_{0}} = & \sigma _{s}^{2}{{\mathbf{h}}_1}\mathbf{h}_1^{H} + \sigma _{n}^{2}{{\mathbf{I}}_{M}}. \label{define_cov2}
	\end{align}
	
	Under hypothesis $\mathcal{H}_1$, we first derive the following theorem on $\mathbf{R}_1$.
	\begin{theorem}
		\label{the1}
		When $\Delta \gamma \rightarrow 0$, the determinant of $\mathbf{R}_1$ is approximated as
		\begin{align}
			\det ({{\mathbf{R}}_{1}}) \approx & \left( \sigma _{s}^{2}{{\left\| {{\mathbf{h}}_{1}} 	\right\|}^{2}}+\sigma _{n}^{2} \right) \left( \left( \frac{1}{k}-\frac{1}{{{k}^{2}}} \right)\sigma _{s}^{2}{{\left\| {{\mathbf{h}}_{2}} \right\|}^{2}}+\sigma _{n}^{2} \right)   \nonumber\\ 
			& \cdot{{(\sigma _{n}^{2})}^{M-2}}.
		\end{align}
	\end{theorem}
	\begin{proof}
	Under hypothesis $\mathcal{H}_1$, the following matrix is first constructed
	\begin{equation}\label{const_matrix1}
		{{\mathbf{\bar{R}}}_{1}}=\left[ \begin{matrix}
			1 & 0 & {({\mathbf{h}}_1 + {\mathbf{h}}_2) ^ H}  \\
			0 & 1 & {{\mathbf{h}} _ 1 ^ H}  \\
			\mathbf{0}_{M\times 1} & \mathbf{0}_{M\times 1} & {{\mathbf{R}}_{1}}
		\end{matrix} \right].
	\end{equation}
	
	It is inferred from \eqref{define_cov} and \eqref{const_matrix1} that $\det ({{\mathbf{\bar{R}}}_{1}})=\det ({{\mathbf{R}}_{1}})$. Then, multiplying the following lower triangular matrices
	\begin{equation}
	\mathbf{A} = \left[ \begin{matrix}
	   1 & 0 & \mathbf{0}_{1\times M}  \\
	   0 & 1 & \mathbf{0}_{1\times M}  \\
	   -\frac{\sigma _{s}^{2}({{\mathbf{h}}_{1}}+{{\mathbf{h}}_{2}})}{k} & -\frac{(k-1)\sigma _{s}^{2}{{\mathbf{h}}_{1}}}{k} & {{\mathbf{I}}_{M}}  \\
	\end{matrix} \right]
	\end{equation}
	and
	\begin{equation}
	\mathbf{B} = \left[ \begin{matrix}
		1 & 0 & -\frac{{{({{\mathbf{h}}_{1}}+{{\mathbf{h}}_{2}})}^{H}}}{\sigma_n^2}  \\
		0 & 1 & -\frac{\mathbf{h}_{1}^{H}}{\sigma_n^2}  \\
		\mathbf{0}_{M\times 1} & \mathbf{0}_{M\times 1} & {{\mathbf{I}}_{M}}  \\
	\end{matrix} \right]
	\end{equation}
	by ${\mathbf{\bar{R}}}_{1}$, namely, $\mathbf{B}\mathbf{A}{{\mathbf{\bar{R}}}_{1}}$, yields
	\begin{equation}
	\label{tansform_matrix_final}
	\begin{aligned}
	\left[ \begin{matrix}
			\frac{\sigma _{s}^{2}{{\left\| {({\mathbf{h}} _ 1 + {\mathbf{h}} _ 2)} \right\|}^{2}}}{k\sigma _{n}^{2}}+1 & \frac{(k - 1)\sigma _{s}^{2}{({\mathbf{h}} _ 1 + {\mathbf{h}} _ 2)^{H}}{{\mathbf{h}}_1}}{k\sigma _{n}^{2}} & \mathbf{0}_{1\times M}  \\
			\frac{\sigma _{s}^{2}{{\mathbf{h}}_1^{H}}{({\mathbf{h}} _ 1 + {\mathbf{h}} _ 2)}}{k\sigma _{n}^{2}} & \frac{(k - 1)\sigma _{s}^{2}{{\left\| {{\mathbf{h}} _1} \right\|}^{2}}}{k\sigma _{n}^{2}}+1 & \mathbf{0}_{1\times M}  \\
			- \frac{\sigma _{s}^{2} {({\mathbf{h}}_1 + {\mathbf{h}}_2)}}{k} & -\frac{(k - 1)\sigma _{s}^{2}{{\mathbf{h}} _ 1}}{k} & \sigma _{n}^{2}{{\mathbf{I}}_{M}}  \\
		\end{matrix} \right].
	\end{aligned}
	\end{equation}
	
	Based on the transformations of $\mathbf{R}_1$ from \eqref{define_cov} to \eqref{tansform_matrix_final}, the determinant of $\mathbf{R}_1$ is obtained in \eqref{func48}, at the bottom of this page. When $\Delta \gamma$ approaches zero and $\mathbf{h}_1$ and $\mathbf{h}_2$ are independent of each other, the value of P2 in (\ref{func48}) is significantly smaller than that of P1. Consequently, Theorem \ref{the1} is proved.
	\end{proof}

	Without loss of generality, we first assume that ${{l}_{1}}\ge {{l}_{2}}\ge ...\ge {{l}_{M}}$ and ${{\lambda }_{1}}\ge {{\lambda }_{2}}\ge ...\ge {{\lambda }_{M}}$ are the eigenvalues of ${{\mathbf{R}_1}}$ and the received signal covariance matrix ${\mathbf{\hat{R}}}$, which can be expressed as
	
	\setcounter{equation}{22}

	\begin{equation}\label{sample-cov}
	{{\mathbf{\hat{R}}}} = \frac{\mathbf{Y}{{\mathbf{Y}}^{H}}}{2N}.
	\end{equation}
	
	As $\Delta \gamma$ approaches zero, according to Theorem \ref{the1}, we have 
	\begin{align}
	{{l}_{1}} & =\sigma _{s}^{2}{{\left\| {{\mathbf{h}}_{1}} \right\|}^{2}}+\sigma _{n}^{2}, \label{hyper_H1_l1}\\
	{{l}_{2}} & =\sigma _{s}^{2}{{\left\| {{\mathbf{h}}_{2}} \right\|}^{2}}\left(\frac{1}{k} - \frac{1}{k^2}\right)+\sigma _{n}^{2}, \label{hyper_H1_l2}\\
	l_3 & = ... = l_M = \sigma_n ^2. \label{hyper_H1_l3}
	\end{align}

	Eqs. \eqref{hyper_H1_l1} and \eqref{hyper_H1_l2} reveal that IDASK, i.e., $k >1$, is beneficial for redirecting the energy of the received signal in the backscatter link towards the second largest eigenvalue of $\mathbf{R}_1$. Thus, the DLI corresponding to the largest eigenvalue can be eliminated from the received AmBC signal. On the contrary, when $k=1$, IDASK becomes the conventional OOK in \cite{related_semi3}.
	
	Under hypothesis $\mathcal{H}_0$, we have $\det ({{\mathbf{R}}_{0}})=(\sigma _{s}^{2}{{\left\| {{\mathbf{h}}_1} \right\|}^{2}}+\sigma _{n}^{2}){{(\sigma _{n}^{2})}^{M-1}}$\cite{ref23}, and
	\begin{align}
	l_1 & = \sigma _{s}^{2}{{\left\| {{\mathbf{h}}_1} \right\|}^{2}}+\sigma _{n}^{2}, \label{hyper_H0_l1} \\
	{{l}_{2}} & =...={{l}_{M}}=\sigma _{n}^{2}. \label{hyper_H0_l2}
	\end{align}
	
	It can be seen that the second largest eigenvalue of $\mathbf{R}_0$ is only related to noise, while the left eigenvalues are identical to those obtained in hypothesis $\mathcal{H}_1$. Therefore, under hypotheses $\mathcal{H}_0$ and $\mathcal{H}_1$, the second largest eigenvalue of the covariance matrix of the received signal can be served as a key indicator for the detection of reflected signals through the BD. Based on this, we design the following detector
	\begin{equation}
		\label{func7}
		T(\mathbf{Y})=\frac{{{\lambda }_{2}}}{\sigma _{n}^{2}}\underset{{{\mathcal{H}}_{0}}}{\overset{{{\mathcal{H}}_{1}}}{\mathop{\gtrless }}}\,\eta ,
	\end{equation}
	where $\eta $ is a given threshold. To conduct the second largest eigenvalue-based symbol detection, the noise variance should be estimated in advance.
	\subsection{Estimation of Noise Variance}
	Under hypothesis $\mathcal{H}_0$, according to \eqref{hyper_H0_l2}, the noise variance is estimated as
	\begin{equation}
		\label{func9}
		\hat{\sigma }_{n}^{2}=\frac{\text{tr}(\mathbf{\hat{R}})-{{\lambda }_{1}}}{M-1},{{\mathcal{H}}_{0}}.
	\end{equation}
	
	 Under hypothesis $\mathcal{H}_1$, according to \eqref{hyper_H1_l3}, the noise variance can be estimated as
	\begin{equation}
		\label{func10}
		\hat{\sigma }_{n}^{2}=\frac{\text{tr}(\mathbf{\hat{R}})-{{\lambda }_{1}}-{{\lambda }_{2}}}{M-2},{{\mathcal{H}}_{1}}.
	\end{equation}
	
	It can be observed from \eqref{func9} and \eqref{func10} that due to the effect of the backscatter link, different expressions are used to estimate the noise variance. In practical blind detection, to accurately evaluate the noise variance, the reflection state of the BD should be first determined. Fortunately, considering the higher second largest value in the reflection state compared to the non-reflection state, ${{\lambda }_{2}}-{{\lambda }_{m}} $ $(m=3,4,\ldots,M)$ will be small in hypothesis $\mathcal{H}_0$, and will be a larger value compared to the noise variance in hypothesis $\mathcal{H}_1$. Thus, the estimation of the noise variance is devised as

	\begin{equation}
	\label{modified_noise}
		\hat{\sigma }_{n}^{2}=\left\{ \begin{aligned}
			& \frac{\text{tr}(\mathbf{\hat{R}})-{{\lambda }_{1}}}{M-1},{{\lambda }_{2}}-{{\lambda }_{m}}<\frac{\text{tr}(\mathbf{\hat{R}})-{{\lambda }_{1}}-{{\lambda }_{2}}}{M-2}, \\ 
			& \frac{\text{tr}(\mathbf{\hat{R}}) \! - \! {{\lambda }_{1}} \! - \! {{\lambda }_{2}}}{M-2},{{\lambda }_{2}}-{{\lambda }_{m}}\ge \frac{\text{tr}(\mathbf{\hat{R}})-{{\lambda }_{1}}-{{\lambda }_{2}}}{M-2}.
		\end{aligned} \right.
	\end{equation}
	
	Consequently, based on \eqref{modified_noise}, the blind detector in \eqref{func7} is improved as
	\begin{equation}
		T(\mathbf{Y}) \! = \! \left\{ \begin{aligned}
			& \frac{{{\lambda }_{2}}\left( M-1 \right)}{\text{tr}( {\mathbf{\hat{R}}} )-{{\lambda }_{1}}},{{\lambda }_{2}}-{{\lambda }_{m}}<\frac{\text{tr}( {\mathbf{\hat{R}}} ) \! - \! {{\lambda }_{1}} \! - \! {{\lambda }_{2}}}{M-2}, \\ 
			& \frac{{{\lambda }_{2}}\left( M-2 \right)}{\text{tr}( {\mathbf{\hat{R}}} ) \!- \! {{\lambda }_{1}} \! - \! {{\lambda }_{2}}},{{\lambda }_{2}} \! - \! {{\lambda }_{m}}\ge \frac{\text{tr}( {\mathbf{\hat{R}}} ) \! - \! {{\lambda }_{1}} \! - \! {{\lambda }_{2}}}{M-2}.
		\end{aligned} \right.
	\end{equation}
	\section{Analytical Performance Evaluation} \label{analyse}
	In this section, we evaluate the detection performance of the proposed blind detector in terms of the false alarm probability, the missed detection probability, and the lower BER bound.	
	
	\subsection{False Alarm Probability}
	Under hypothesis $\mathcal{H}_0$, when threshold $\eta$ is much lower than the value of $\lambda _2 / \sigma^2_n$, the received AmBC signal is susceptible to being erroneously detected as a signal with the reflected RF signal, resulting in a false alarm. To quantify the impact of the threshold on the false alarm, we analyze the false alarm probability. Firstly, we define the ratio of the transmitted RF signal variance to the channel noise variance as $\gamma = \sigma^2_{s}/\sigma^2_{n}$. Then, under the assumption of $N \gg M$, the following lemma is proposed to demonstrate the distribution of $\lambda _2 / \sigma^2_n$.
	\begin{lemma}
		\label{lem1}
		When $N \gg M$ and $\gamma$ is large, the second largest eigenvalue $\lambda _2$ normalized by $\sigma^2_n$ follows
		\begin{equation}
			\label{func19}
			\frac{{{\lambda }_{2}}/\sigma _{n}^{2}-{{\mu }_{N,M-1}}}{{{\sigma }_{N,M-1}}}\sim T{{W}_{2}},
		\end{equation}
		where $TW_2$ denotes the Tracy-Widom distribution of order 2\cite{ref26}, and
		\begin{align}
			& {{\mu }_{N,M-1}}={{\left( 1+\sqrt{\frac{M-1}{2N}} \right)}^{2}}, \\ 
			& {{\sigma }_{N,M-1}} \! = \! \frac{1}{\sqrt{2N}}\left(1 \! + \! \sqrt{\frac{M-1}{2N}} \right){{\left( \frac{1}{\sqrt{2N}} \! + \! \frac{1}{\sqrt{M-1}} \right)}^{1/3}}. 
		\end{align}
	\end{lemma}
	\begin{proof}
	See Appendix \ref{appen 1}.
	\end{proof}

	Then, according to Lemma \ref{lem1}, the false alarm probability $P_{\rm{fa}}$ is written as
		\begin{align}
		\label{theory_Pfa}
			{{P}_{\rm{fa}}} & =P\left[ T(\mathbf{Y})>\eta |{{\mathcal{H}}_{\text{0}}} \right] \nonumber\\ 
			& =P\left[ \frac{{{\lambda }_{2}}/\sigma _{n}^{2}-{{\mu }_{N,M-1}}}{{{\sigma }_{N,M-1}}} > \left. \frac{\eta -{{\mu }_{N,M-1}}}{{{\sigma }_{N,M-1}}}  \right| {{\mathcal{H}}_{0}}  \right] \nonumber\\ 
			& =1-{{F}_{TW2}}\left( \frac{\eta -{{\mu }_{N,M-1}}}{{{\sigma }_{N,M-1}}} \right), 
		\end{align}
	where ${{F}_{TW2}}\left( \cdot  \right)$ denotes the cumulative distribution function (CDF) of the Tracy-Widom distribution of order 2.

	For a given $P_{\rm{fa}}$, the threshold is derived as
	\begin{equation}
		\label{func20}
		\eta ={{\mu }_{N,M-1}}+{{\sigma }_{N,M-1}}F_{TW2}^{-1}(1-{{P}_{\rm{fa}}}),
	\end{equation}
	where ${{F}_{TW2}^{-1}}\left( \cdot  \right)$ denotes the inverse function of ${{F}_{TW2}}\left( \cdot  \right)$.
	
	Fig. \ref{P_fa_vs_MN} illustrates the analytical false alarm probability $P_{\rm{fa}}$ in \eqref{theory_Pfa} with different values of $M$, $N$, and $\eta$. It shows that an increase in $N$ when $N \gg M$, or an increase in the threshold value of $\eta$ can result in a notable decrease in the false alarm probability $P_{\rm{fa}}$. 
	
	\begin{figure}[h]
		\centering
		\includegraphics[width=3.5in]{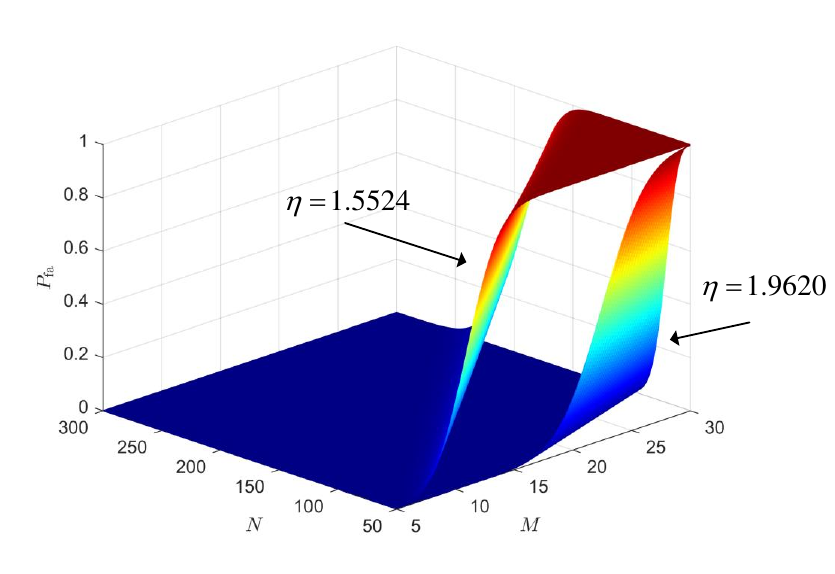}
		\caption{Analytical false alarm probability with varying values of $M$, $N$, and $\eta$, where $M = 5, 6, \ldots, 30 $, $N = 50, 60, \ldots, 300$, and $\eta = 1.5524,1.962$.}
		\label{P_fa_vs_MN}
	\end{figure}
	\subsection{Missed Detection Probability}
	Under hypothesis ${{\mathcal{H}}_{1}}$, when threshold $\eta$ is much larger than the value of $\lambda _2 / \sigma^2_n$, the reflected RF signal is susceptible to being missed. Thus, we analyze the missed detection probability about the proposed blind detector. Firstly, according to Theorem \ref{the1}, \eqref{hyper_H1_l1} and \eqref{hyper_H1_l2}, the distribution of $\lambda _2 $ is given in the following lemma.
	\begin{lemma}
		\label{lem2}
		When $N \gg M$, $\Delta \gamma$ is low, and $\gamma {{\left\| {{\mathbf{h}}_{2}} \right\|}^{2}} \ge 1$, the distribution of $\lambda_{2} / \sigma _n ^2$ follows
		\begin{equation}
			\label{func23}
			\frac{{{\lambda }_{2}}}{\sigma _{n}^{2}}\sim \mathcal{N}\left( (1+{{\gamma }_{1}})\left( 1+\frac{M-2}{2N{{\gamma }_{1}}} \right),\frac{{{(1+{{\gamma }_{1}})}^{2}}}{2N} \right),
		\end{equation}
		where ${{\gamma }_{1}}={{\left\| {{\mathbf{h}}_{2}} \right\|}^{2}} \left(1/k - 1/k^2\right) \gamma$ .
	\end{lemma}
	\begin{proof}
		See Appendix \ref{appen 3}.
	\end{proof}
	According to Lemma \ref{lem2}, the missed detection probability $P_{\rm{md}}$ is expressed as
	\begin{align}\label{missed_detection}
		{{P}_{\rm{md}}} & =P\left[ Y(\mathbf{Y})<\eta |{{\mathcal{H}}_{1}} \right] \nonumber\\ 
		& =1-\mathcal{Q}\left( \frac{\eta -\left( 1+{{\gamma }_{1}} \right)\left( 1+\frac{M-2}{2N{{\gamma }_{1}}} \right)}{\frac{1+{{\gamma }_{1}}}{\sqrt{2N}}} \right).
	\end{align}
	
	It can be concluded from \eqref{missed_detection} that the missed detection probability is the conditional probability of $\gamma_1$, which is time-varying in the same way as ${{\mathbf{h}}_{2}}$. We assume that ${{f}_{{{\gamma }_{1}}}}(x)$ is the PDF of ${{\gamma }_{1}}$. Hence, the average probability of missed detection can be calculated by
	\begin{align}
	\label{theory_Pmd}
		{\overline{{{P}_{\rm{md}}}}} & =\int_{0}^{\infty }{{{P}_{\rm{md}}}(x){{f}_{{{\gamma }_{1}}}}(x)dx} \nonumber\\ 
		& = \! 1 \! - \! \int_{0}^{\infty }{\mathcal{Q}\left( \frac{\eta \! - \! \left( 1 \! + \! x \right)\left( 1+\frac{M-2}{2Nx} \right)}{\frac{1+x}{\sqrt{2N}}} \right){{f}_{{{\gamma }_{1}}}}(x)dx}.
	\end{align}
	
	\subsection{Bit Error Rate}
	According to \eqref{theory_Pfa} and \eqref{theory_Pmd}, the BER can be written as
	\begin{align}
	P_e = & \frac{P_{\rm{fa}} + \overline{{{P}_{\rm{md}}}}}{2}  \nonumber \\
	= & 1 - \frac{1}{2} \int_{0}^{\infty }{\mathcal{Q}\left( \frac{\eta -\left( 1+x \right)\left( 1+\frac{M-2}{2Nx} \right)}{\frac{1+x}{\sqrt{2N}}} \right){{f}_{{{\gamma }_{1}}}}(x)dx} \nonumber \\
	& - \frac{1}{2} {{F}_{TW2}}\left( \frac{\eta -{{\mu }_{N,M-1}}}{{{\sigma }_{N,M-1}}} \right).
	\end{align}
	It is difficult to derive the closed-form expression of the BER. Therefore, we analyze a lower bound of the BER for the blind detector. We assume $P_0$ and $P_1$ denote the probability distributions of ${{\lambda }_{2}}/{\sigma _{n}^{2}}$ under the hypotheses of $\mathcal{H}_0$ and  $\mathcal{H}_1$, respectively. The total variation \cite{bash} can be formulated as
	\begin{equation}
		\mathcal{V}\left( {{P}_{0}}\left\| {{P}_{1}} \right. \right)\triangleq \frac{1}{2}{{\left\| {{p}_{0}}\left( x \right)-{{p}_{1}}\left( x \right) \right\|}_{1}},
	\end{equation}
	 where $p_0$ and $p_1$ denote the probability densities of $P_0$ and $P_1$, respectively, and ${{\left\| \cdot  \right\|}_{1}}$ refers to the $\mathcal{L}_1$ norm. According to \cite{ref28}, the BER of the blind detector satisfies
	\begin{equation}\label{BER}
		{{P}_{e}}\ge 1-\mathcal{V}\left( {{P}_{0}}||{{P}_{1}} \right),
	\end{equation}
	where $\mathcal{V}\left( {{P}_{0}}||{{P}_{1}} \right)$ is characterized by the following theorem.
	\begin{theorem}\label{th2}
	When $N \gg M$ and $\Delta \gamma$ is low, as $\gamma$ approaches infinity, the total variation between $P_0$ and $P_1$ converges to
	\begin{equation}
		\underset{\gamma\to \infty }{\mathop{\lim }}\,\mathcal{V}\left( {{P}_{0}}||{{P}_{1}} \right) \approx \frac{1}{2}\left(1 + \mathcal{Q}\left(-\sqrt{2N}\right)\right).
	\end{equation}
	\end{theorem}
	\begin{proof}
	The total variation between $P_0$ and $P_1$ is written as
	\begin{align}
  		\underset{\gamma\to \infty }{\mathop{\lim }}\,\mathcal{V}\left( {{P}_{0}}\left\| {{P}_{1}} \right. \right) = & \underset{\gamma\to \infty }{\mathop{\lim }}\, \frac{1}{2} \int_{0}^{\infty }{\left| {{p}_{0}}\left( x \right)-{{p}_{1}}\left( x \right) \right|}dx \nonumber\\ 
 		= & \frac{1}{2} \int_{0}^{a}{\left| {{p}_{0}}\left( x \right)-\underset{\gamma\to \infty }{\mathop{\lim }}\,{{p}_{1}}\left( x \right) \right|}dx \nonumber\\
 		& + \frac{1}{2} \int_{a}^{\infty }{\left| {{p}_{0}}\left( x \right)-\underset{\gamma\to \infty }{\mathop{\lim }}\,{{p}_{1}}\left( x \right) \right|}dx, 
	\end{align}
	where $a$ is a non-negative large number. Through numerical calculations, for a large value of $a$, we have $\int_{a}^{\infty }{{{p}_{0}}\left( x \right)dx}\approx 0$. In this case, the probability of $p_1(x)$ in the interval $[a, +\infty]$ is expressed as
	\begin{align}
	\label{func28}
		\underset{\gamma\to \infty }{\mathop{\lim }}\,\int_{a}^{\infty }{{{p}_{1}}\left( x \right)dx} & \! = \! \underset{\gamma\to \infty }{\mathop{\lim }}\,\mathcal{Q}\left( \frac{a \! - \!\left( 1+{{\gamma }_{1}} \right)\left( 1+\frac{M-2}{2N{{\gamma }_{1}}} \right)}{\frac{\left( 1+{{\gamma }_{1}} \right)}{\sqrt{2N}}} \right) \nonumber\\ 
		& = \mathcal{Q}\left(-\sqrt{2N}\right),
	\end{align}
	where ${{\gamma }_{1}}={{\left\| {{\mathbf{h}}_{2}} \right\|}^{2}}(1/k - 1/k^2)\gamma$. Consequently, we have
	\begin{align}
		\underset{\gamma\to \infty }{\mathop{\lim }}\,\mathcal{V}\left( {{P}_{0}}\left\| {{P}_{1}} \right. \right) & \approx \frac{1}{2} \int_{0}^{a}{{{p}_{0}}\left( x \right)dx} \! + \! \underset{\gamma\to \infty }{\mathop{\lim }}\, \frac{1}{2} \int_{a}^{\infty }{{{p}_{1}}\left( x \right)dx} \nonumber\\
		& \approx \frac{1}{2}\left(1 + \mathcal{Q}\left(-\sqrt{2N}\right)\right).
	\end{align}
	
	Theorem \ref{th2} is proved.
	\end{proof}
	
	Based on Theorem \ref{th2}, the lower bound of the BER in \eqref{BER} can be expressed as
	\begin{equation}
	\underset{\gamma\to \infty }{\mathop{\lim }}\,P_e \ge \frac{1}{2}\left(1 - \mathcal{Q}\left(-\sqrt{2N}\right) \right).
	\end{equation}
	
	When $N \rightarrow +\infty$, the lower bound of $ P_e $ approaches 0. It implies that increasing the value of $N$ can reduce the BER of the proposed blind detector. Furthermore, according to Theorem \ref{th2}, it is evident that an increase in ${{\gamma }_{1}}={{\left\| {{\mathbf{h}}_{2}} \right\|}^{2}} \left(1/k - 1/k^2\right)\gamma$ is directly associated with $\gamma$ and the IDASK. Given a fixed value of $\gamma$, the maximum value of ${{\gamma }_{1}}$ is achieved when $k=2$. In this case, the proposed blind detector (\ref{func7}) exhibits the optimal detection performance.
	\begin{figure*}[t]
	\centering	
	\begin{minipage}{0.45\linewidth}
		\centering
		\includegraphics[width=3.5in]{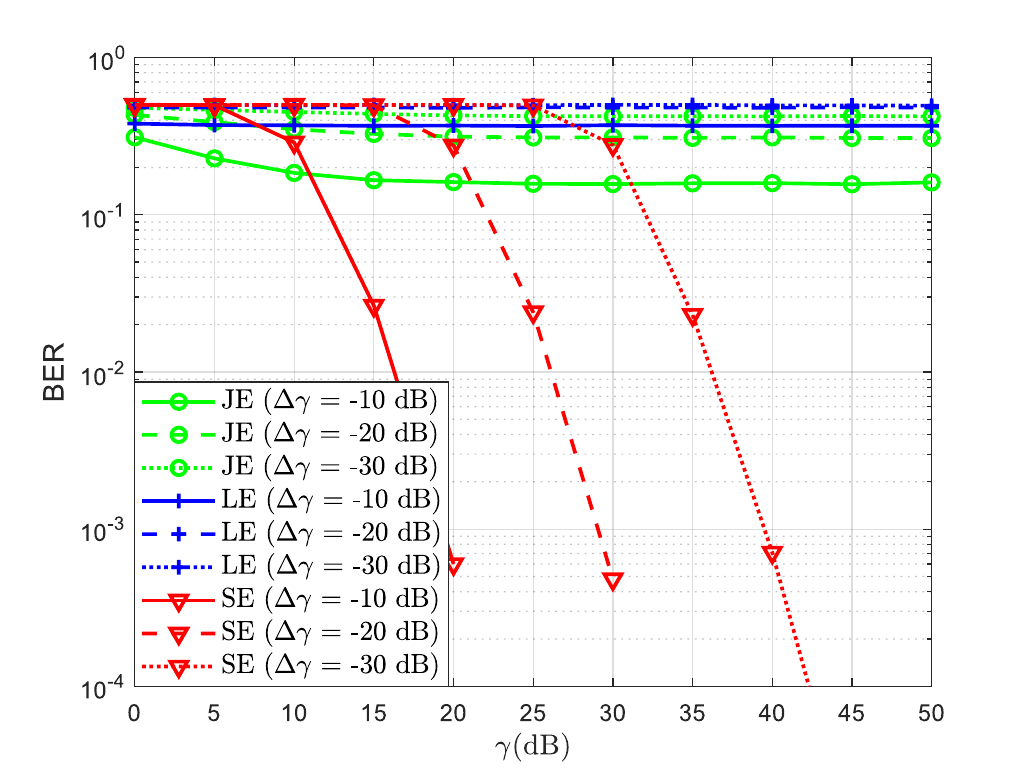}
		\caption{BER versus $\gamma$ for the JE, LE and SE detectors, where $\Delta \gamma = -10,-20,-30$ dB.}
		\label{fig4}
	\end{minipage}
	\hspace{4mm}
	\begin{minipage}{0.45\linewidth}
		\centering
		\includegraphics[width=3.5in]{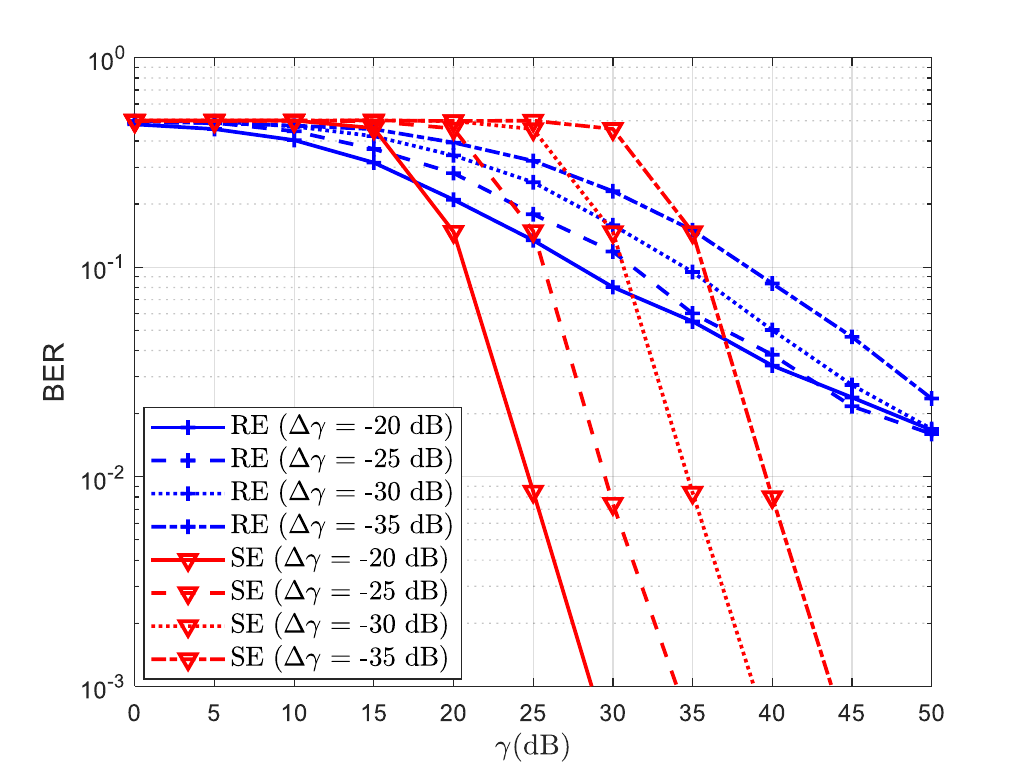}
		\caption{BER versus $\gamma$ for the RE and SE detectors, where $\Delta \gamma = -20,-25,-30,-35$ dB.}
		\label{fig13}
	\end{minipage}
	\end{figure*}
	\begin{figure}[t]
		\centering
		\includegraphics[width=3.5in]{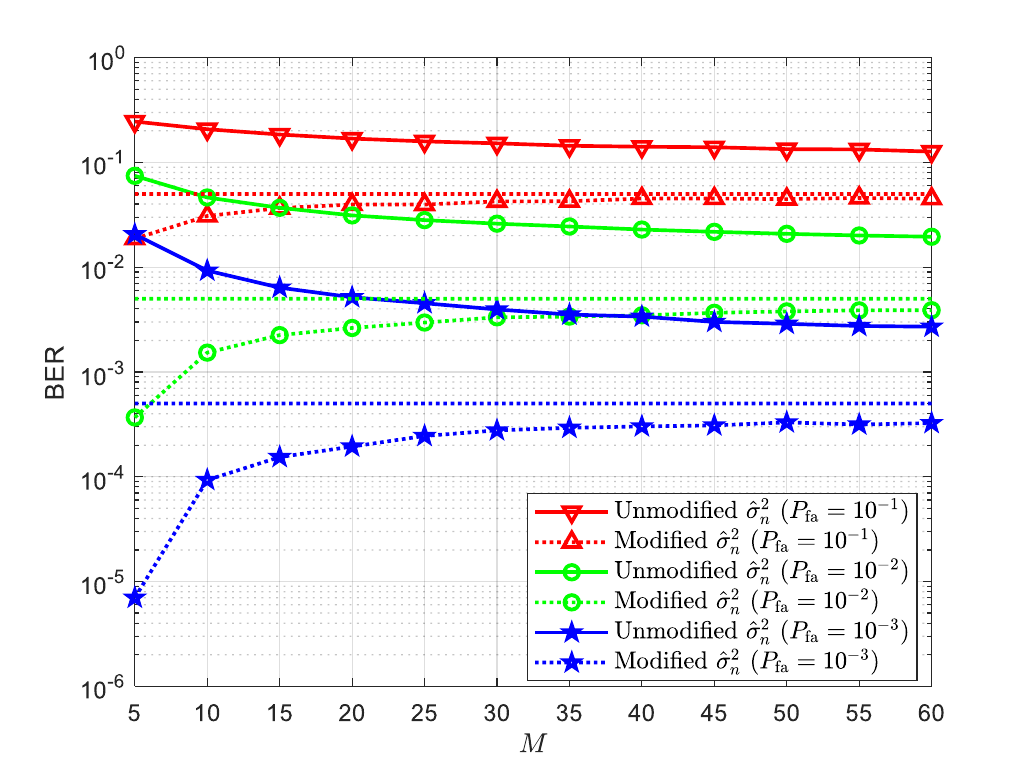}
		\caption{BER versus $M$ for the SE detector under different $P_{\rm{fa}}$ and different estimation schemes of $\hat{\sigma }_{n}^{2}$.}
		\label{fig10}
	\end{figure}
	\section{Simulation Results} \label{simulation}
	In this section, we present simulation results to demonstrate the detection performance of the proposed blind detector. In the presented simulation results, we assume that the noise variance is $\sigma _{n}^{2}=-20 $ dBm, the direct channel gain is $0$ dB, and the value of $m$ in \eqref{modified_noise} is $m=3$. The Monte Carlo method is used to calculate the BER and $P_{\rm{md}}$. For fair comparison, we compare our proposed second largest eigenvalue (SE) based detector with the joint-energy (JE) detector in\cite{related_non3}, the largest eigenvalue (LE) based detector in \cite{related_semi3}, and the efficient ratio (RE) detector in \cite{related_semi5} under the same simulation condition.	
	
	\cref{fig4,fig13} plot the BER versus $\gamma$ of the JE, LE, RE, and SE detectors with different values of $\Delta \gamma$. The number of antennas at the receiver is set to $M=5$ and the BD symbol period is $2N=100$. In the JE and LE detectors, the accurate statistical variances of the received signals are used for the calculation of the detection threshold. In the RE detector, the perfect CSI estimation is assumed. In the SE detector, the detection threshold $\eta$ is employed with $P_{\rm{fa}} = 10^{-4}$ and $P_{\rm{fa}} = 10^{-2}$ in \cref{fig4} and \cref{fig13}, respectively. It can be observed that as the value of $\Delta \gamma$ increases, the BERs of all detectors decrease. Among them, the proposed SE detector shows the best BER performance. Furthermore, with an increase of $\gamma$, the BERs of the SE detector are significantly reduced compared to those of the JE, LE, and RE detectors.
	
	\cref{fig10} shows the BER versus $M$ of the SE detector for a large value of $\gamma$, where $P_{\rm{fa}} = 10^{-1}, 10^{-2}, 10^{-3}$. The noise variance is estimated by the unmodified scheme in \eqref{func10} and the modified scheme in \eqref{modified_noise}. We set $\gamma = 50$ dB, $\Delta \gamma = -20$ dB, and $N = 10M$. As the values of $M$ increases, all the BER curves converge to $P_{\rm{fa}}/2$. More importantly, the BER performance of the modified noise estimation scheme outperforms that of the unmodified scheme.
	\begin{figure*}[t]
	\centering
	\begin{minipage}{0.45\linewidth}
		\centering
		\includegraphics[width=3.5in]{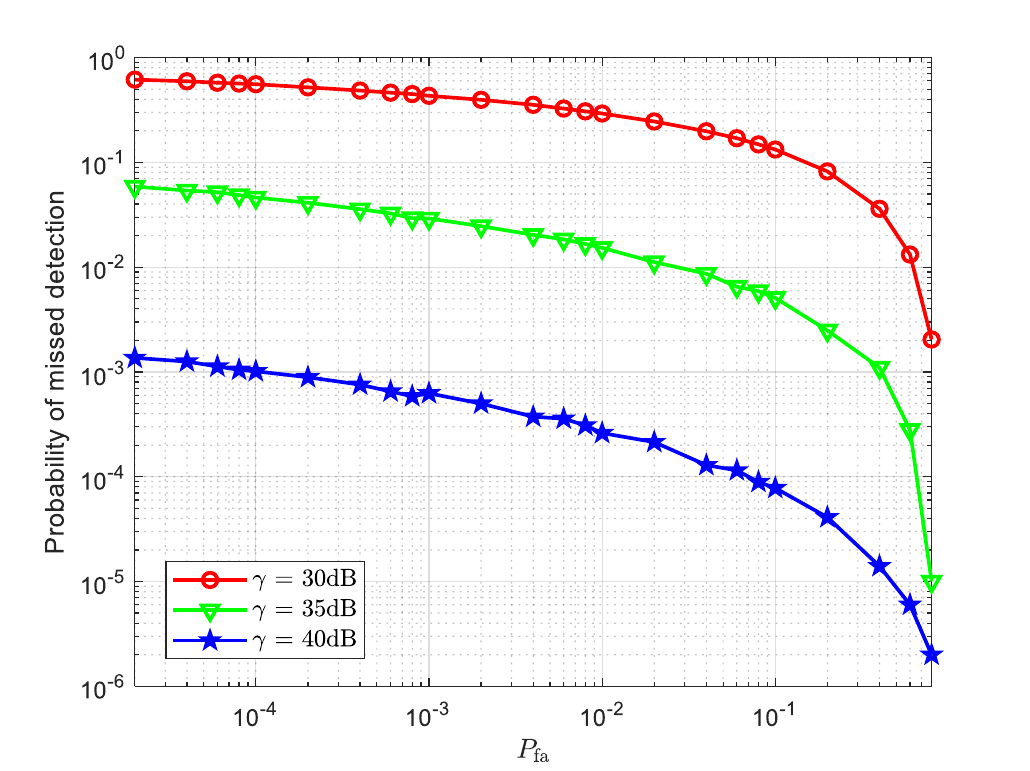}
		\caption{Complementary ROC ($P_{\rm{md}}$ vs. $P_{\rm{fa}}$) for the SE detector, where $\gamma = 30,35,40$ dB.}
		\label{fig5}
	\end{minipage}
	\hspace{4mm}
	\begin{minipage}{0.45\linewidth}
		\centering
		\includegraphics[width=3.5in]{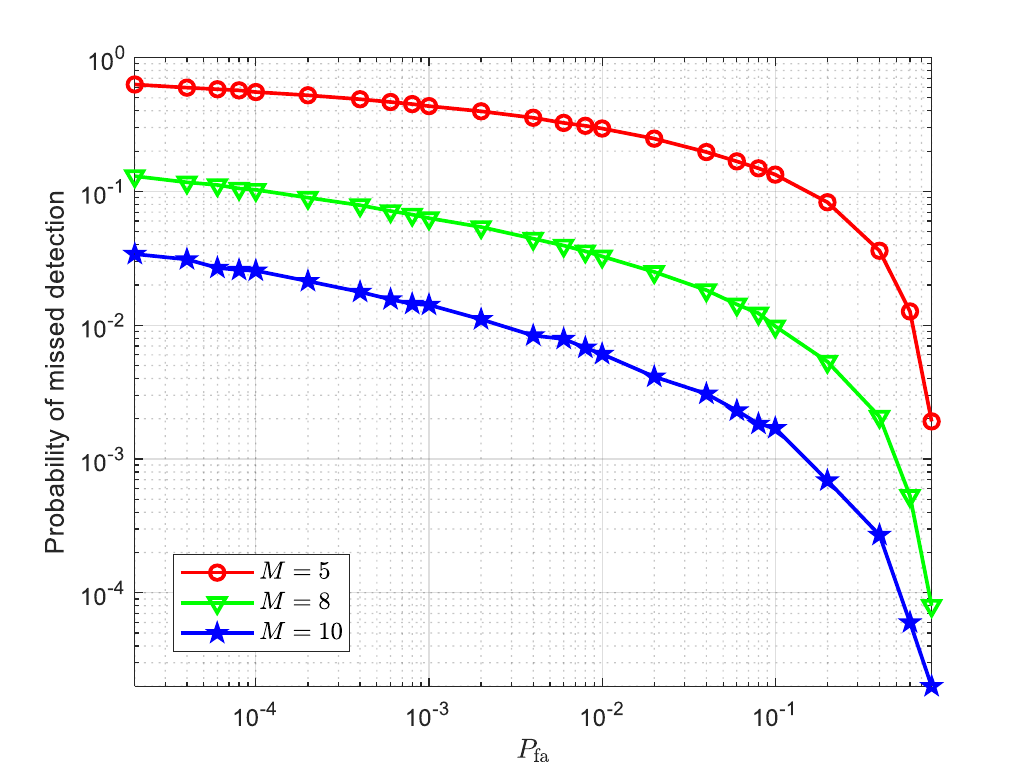}
		\caption{Complementary ROC of the SE detector, where $M = 5,8,10$.}
		\label{fig6}
	\end{minipage}
	\begin{minipage}{0.45\linewidth}
		\centering
		\includegraphics[width=3.5in]{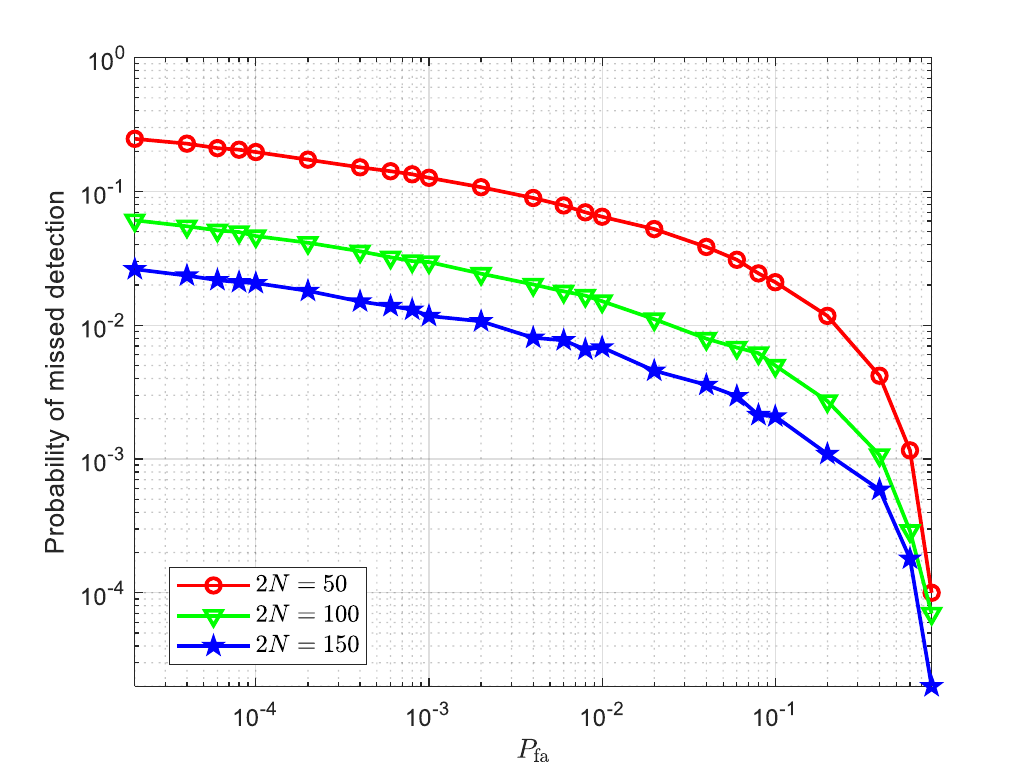}
		\caption{Complementary ROC of the SE detector, where $2N = 50, 100, 150$.}
		\label{fig8}
	\end{minipage}
	\hspace{4mm}
	\begin{minipage}{0.45\linewidth}
		\centering
		\includegraphics[width=3.5in]{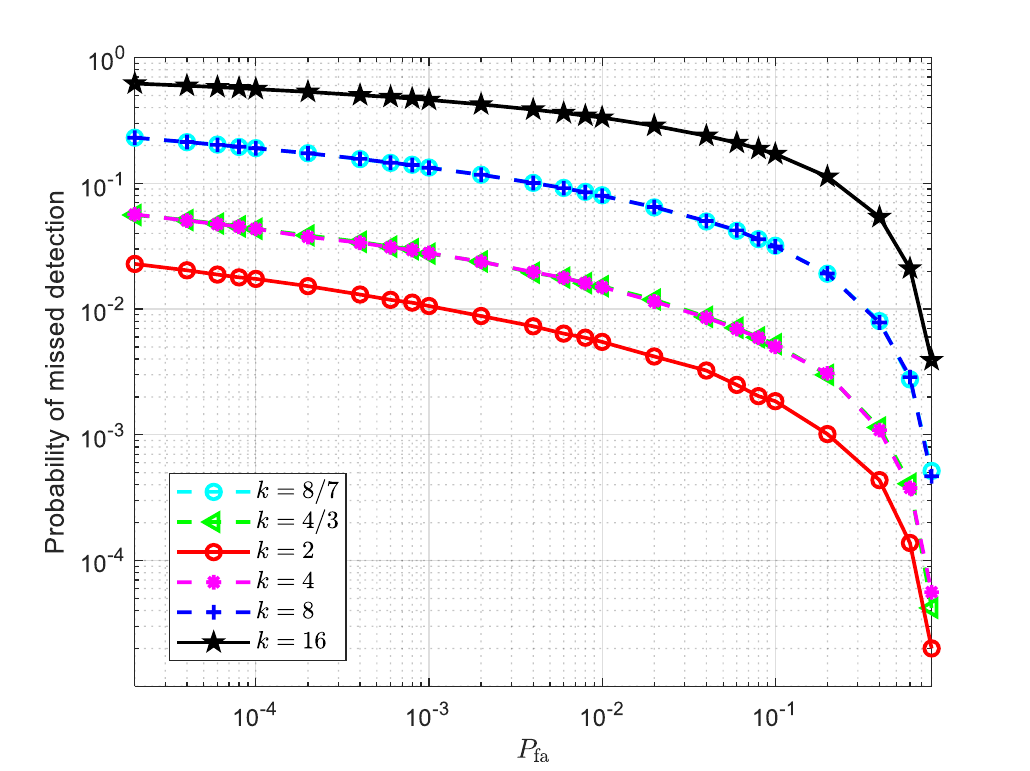}
		\caption{Complementary ROC of the SE detector, where $k = 8/7, 4/3, 2, 4, 8, 16$.}
		\label{fig9}
	\end{minipage}
	\end{figure*}
	
	\cref{fig5,fig6,fig8,fig9} plot the probabilities of missed detection versus $P_{\rm{fa}}$ of the proposed SE detector. Specifically, \cref{fig5} shows the probability of missed detection versus $P_{\rm{fa}}$ under different values of $\gamma$. We set $\Delta \gamma = -30$ dB, $M=5$, and $N=50$. Indeed, the detection performance of SE detector is noticeably improved as $\gamma$ increases. \cref{fig6} shows the probability of missed detection with different $M$, where $\gamma = 30$ dB, $\Delta \gamma = -30$ dB, and $N=50$. It shows that the detection performance of the SE detector is enhanced as $M$ increases, indicating the performance gain provided by spatial diversity. \cref{fig8} demonstrates the probability of missed detection with different values of $N$, where $\gamma = 35$ dB, $\Delta \gamma = -30$ dB and $M=5$. It shows that the detection performance of the SE detector improves as the $N$ increases. Furthermore, \cref{fig9} plots the probability of missed detection of the SE detector with different values of $k$, where $\gamma = 35$ dB, $\Delta \gamma = -30$ dB, $M=5$, and $N=80$. As can be observed in \cref{fig9}, the SE detector exhibits the optimal detection performance when $k = 2$, which is consistent with the analysis in Section IV-B. In addition, these figures also show that the probability of missed detection is significantly reduced by increasing $P_{\rm{fa}}$.
	
	Moreover, \cref{fig11} compares the BER versus $\gamma$ of the SE detector with various digital modulation types in the ambient RF signal. The digital modulation types are BPSK, QPSK, and 16-QAM. The results demonstrate that the proposed SE detector exhibits a similar BER performance for different digital modulation types.
	\begin{figure}[t]
		\centering
		\includegraphics[width=3.5in]{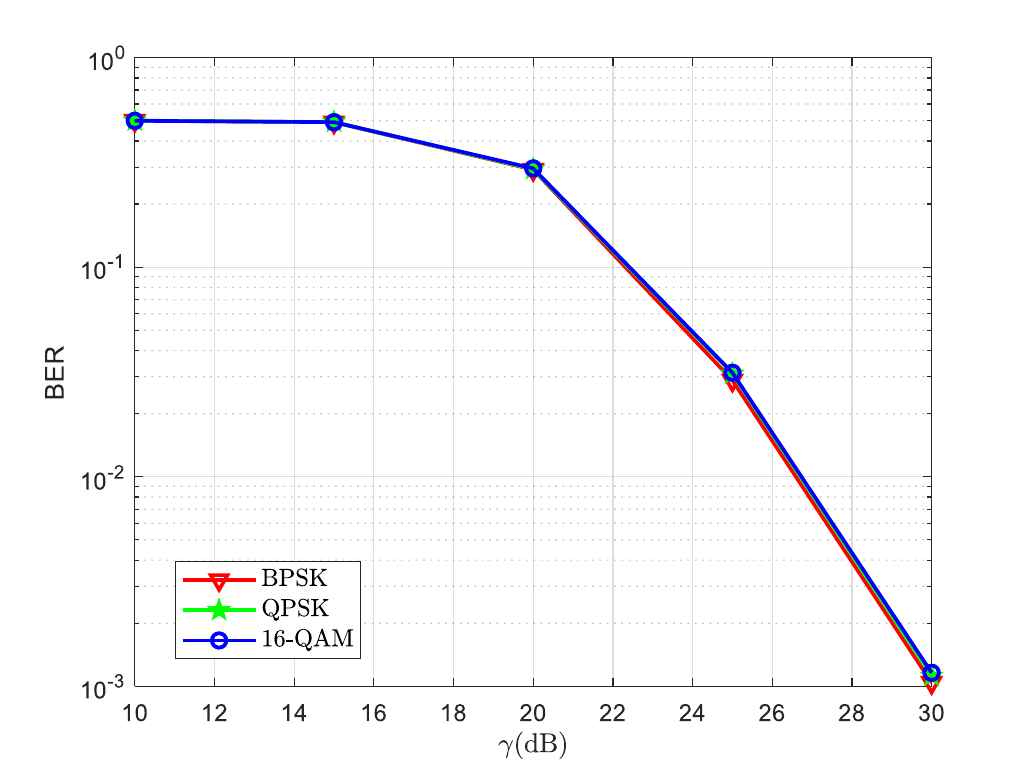}
		\caption{BER versus $\gamma$ of the SE detector. The ambient RF signal employs a variety of digital modulation types, including BPSK, QPSK, and 16-QAM.}
		\label{fig11}
	\end{figure}
	\section{Conclusion} \label{sum}
	This paper investigated the blind symbol detection problem of the AmBC system. Firstly, by leveraging IDASK, the second largest eigenvalue of the received signal covariance matrix was distinguished from the DLI in the proposed blind symbol detector. Moreover, an improved noise estimation scheme was presented to enhance the accuracy of noise variance estimation. Secondly, a theoretical analysis was conducted to evaluate the false alarm probability and the missed detection probability of the blind detector, followed by a lower BER bound. Finally, the simulation results validated that the proposed detector exhibits the optimal detection performance with IDASK using a reflection ratio of $50\%$, i.e., $k = 2$, and exhibits a great detection performance compared to conventional counterparts.
	
	\appendices	
	
	\section{Proof of Lemma \ref{lem1}}
	\label{appen 1}
	We first formulate a binary hypothesis testing when the receiver is equipped with $M-1$ antennas.
	In the presence of a direct channel, the null hypothesis $\mathcal{H}_0$ is applicable. 
	Conversely, in the absence of a direct channel, the alternative hypothesis $\mathcal{H}_1$ is relevant.
	Thus, the received signal $\bar{\mathbf{Y}}=[{{\mathbf{y}}_{1}},{{\mathbf{y}}_{2}},...,{{\mathbf{y}}_{2N}}]\in {{\mathbf{C}}^{(M - 1)\times 2N}}$ follows
	\begin{equation}\label{construct_problem1}
		{\bar{\mathbf{Y}}}\sim \left\{ \begin{aligned}
		& \mathcal{C}\mathcal{N}(0,\bar{\mathbf{R}}_{0}),{{\mathcal{H}}_{0}}, \\ 
		& \mathcal{C}\mathcal{N}(0,\bar{\mathbf{R}}_{1}),{{\mathcal{H}}_{1}}, \\ 
		\end{aligned} \right.
	\end{equation}
	where ${\bar{\mathbf{R}}_{0}}=\sigma _{s}^{2}{{\bar{{\mathbf{h}}}}_{1}}\bar{\mathbf{h}}_{1}^{H}+\sigma _{n}^{2}{{\mathbf{I}}_{M-1}}$, ${\bar{\mathbf{R}}_{1}}= \sigma _{n}^{2}{{\mathbf{I}}_{M - 1}}$, and $\bar{\mathbf{h}}_1$ denotes the direct channel coefficient. 
	We assume that the covariance matrix of $\bar{\mathbf{Y}}$ is $\bar{\mathbf{R}}$, whose eigenvalues satisfy ${\bar{\lambda }_{1}}\ge {\bar{\lambda }_{2}}\ge ...\ge {\bar{\lambda }_{M-1}}$. 
	Under hypothesis $\mathcal{H}_0$, according to \cite{ref25,ref24}, the distribution of the largest eigenvalue of $\bar{\mathbf{R}}$ follows 
	\begin{equation}
		\label{func33}
		{\bar{\lambda }^0_{1}}/\sigma _{n}^{2}\sim \mathcal{N}\left( \left( 1+{{\gamma }_{0}} \right)\left( 1+\frac{M-2}{2N{{\gamma }_{0}}} \right),\frac{{{\left( 1+{{\gamma }_{0}} \right)}^{2}}}{2N} \right),
	\end{equation}
	where $\bar{\lambda }^i_{m}$ is the $m$-th eigenvalue of $\bar{\mathbf{R}}$ under hypothesis $\mathcal{H}_i$ for $m=1,2,\ldots,M-1$ and $i=0,1$, and ${{\gamma }_{0}}={{\left\| {\bar{\mathbf{h}}_1} \right\|}^{2}} \gamma$. Then, under hypothesis $\mathcal{H}_1$, based on \cite{ref29}, the distribution of the largest eigenvalue of $\bar{\mathbf{R}}$ satisfies
	\begin{equation} \label{eq:eq52}
		\frac{{\bar{{\lambda }}^1_1}/\sigma _{n}^{2}-{{\mu }_{N,M-1}}}{{{\sigma }_{N,M-1}}}\sim T{{W}_{2}}.
	\end{equation}
	
	In the original hypotheses of \eqref{func8}, for hypothesis $\mathcal{H}_0$, $\lambda_{1}$ and $\lambda_{2}$ of $\hat{\mathbf{R}}$ in \eqref{sample-cov} can be recalculated as 
	\begin{equation} \label{eq:eq53}
		\lambda_{1}/\sigma _{n}^{2} = \text{max}\{{\bar{\lambda }^{0}_{1}}/\sigma _{n}^{2},{\bar{\lambda }^{1}_{1}}/\sigma _{n}^{2}\}, 
		\end{equation}
		\begin{equation}\label{eq:eq54}
		\lambda_{2}/\sigma _{n}^{2} = \text{min}\{{\bar{\lambda }^{0}_{1}}/\sigma _{n}^{2},{\bar{\lambda }^{1}_{1}}/\sigma _{n}^{2}\}.
	\end{equation}
	
	The PDFs of ${\bar{\lambda }^{0}_{1}}/\sigma _{n}^{2}$ and ${\bar{\lambda }^{1}_{1}}/\sigma _{n}^{2}$ is denoted by $\bar{p}_0$ and $\bar{p}_1$, respectively.
	According to the definition of the overlapping coefficient (OVL) in \cite{OVL_define}, for any $a \in (-\infty, +\infty)$, the OVL of ${\bar{\lambda}^{0}_{1}}/\sigma^{2}_{n}$ and ${\bar{\lambda}^{1}_{1}} / \sigma_{n}^{2}$ is written as
	\begin{align}\label{ovl1}
	\text{OVL} 
	& = \int^{+\infty}_{-\infty} \min\{\bar{p}_0(x), \bar{p}_1(x)\} dx   \nonumber \\
	& = \int\limits^{+\infty}_{a} \min\{\bar{p}_0(x), \bar{p}_1(x)\} dx  + \int\limits^{a}_{-\infty} \min\{\bar{p}_0(x), \bar{p}_1(x)\} dx   \nonumber \\
	& \le \int_{a}^{\infty }{\bar{p}_{1}}\left( x \right)dx + \int_{-\infty}^{a }{\bar{p}_{0}}\left( x \right)dx.
	\end{align}
	Based on \eqref{func33}, we have
	\begin{align}\label{norm_appro}
		\int_{-\infty}^{a}{{\bar{p}_{0}}\left( x \right)dx} & = 1 - \mathcal{Q}\left( \frac{a \! - \!\left( 1+{{\gamma }_{0}} \right)\left( 1+\frac{M-2}{2N{{\gamma }_{0}}} \right)}{\frac{\left( 1+{{\gamma }_{0}} \right)}{\sqrt{2N}}} \right).
	\end{align}
	As $N \gg M$ and $\gamma_0 \rightarrow +\infty$, \eqref{norm_appro} can be approximated as
	\begin{equation}
	\int_{-\infty}^{a}{{\bar{p}_{0}}\left( x \right)dx} \approx 1 - \mathcal{Q}\left(-\sqrt{2N}\right) \approx 0.
	\end{equation}
	Furthermore, based on \eqref{eq:eq52}, through numerical calculations, we obtain $\int_{0}^{5}{{\bar{p}_{1}}\left( x \right)dx}=1-{{10}^{-10}}$.
	Consequently, for a large value of $a$, such as $a>5$, we have $\int_{a}^{+\infty }{{\bar{p}_{1}}\left( x \right)dx}\approx 0$.
	 
	 Therefore, for $N \gg M$, $\gamma_0 \rightarrow +\infty$, and a large value of $a$, the OVL in \eqref{ovl1} approaches zero.
	 In this case, the distribution of $\lambda_{2}/\sigma _{n}^{2}$ can be approximated by the distribution of ${\bar{\lambda }^{1}_{1}}/\sigma _{n}^{2}$.
	 Thus, Lemma \ref{lem1} is proved.

	\section{Proof of Lemma \ref{lem2}}
	\label{appen 3}
	Similar to Appendix \ref{appen 1}, we also formulate a binary hypothesis testing when the receiver is equipped with $M-1$ antennas.
	In the presence of the direct channel, the null hypothesis $\mathcal{H}_0$ is applicable. 
	In the presence of the cascaded channel, the alternative hypothesis $\mathcal{H}_1$ is relevant.
	Thus, the received signal $\bar{\mathbf{Y}}=[{{\mathbf{y}}_{1}},{{\mathbf{y}}_{2}},...,{{\mathbf{y}}_{2N}}]\in {{\mathbf{C}}^{(M-1)\times 2N}}$ follows
	\begin{equation}\label{construct_problem2}
		\bar{\mathbf{Y}}\sim \left\{ \begin{aligned}
		& \mathcal{C}\mathcal{N}(0,\bar{\mathbf{R}}_{0}),{{\mathcal{H}}_{0}}, \\ 
		& \mathcal{C}\mathcal{N}(0,\bar{\mathbf{R}}_{1}),{{\mathcal{H}}_{1}}, \\ 
		\end{aligned} \right.
	\end{equation}
	where ${\bar{\mathbf{R}}_{0}}\!=\!\sigma _{s}^{2}{{\bar{{\mathbf{h}}}}_{1}}\bar{\mathbf{h}}_{1}^{H}+\sigma _{n}^{2}{{\mathbf{I}}_{M-1}}$, ${\bar{\mathbf{R}}_{1}}\!=\!\sigma _{s}^{2}{{\bar{\mathbf{h}}}_{2}}\bar{\mathbf{h}}_{2}^{H} K + \sigma _{n}^{2}{{\mathbf{I}}_{M-1}}$, $\bar{\mathbf{h}}_1$ and $\bar{\mathbf{h}}_2$ denote the direct and cascaded channels, respectively, and $K \! =\! 1/k\! -\! 1/k^2$. Under hypothesis $\mathcal{H}_i (i\in\{0,1\})$, the distribution of the largest eigenvalue $\bar{\lambda }^{i}_{1}/\sigma _{n}^{2}$ of $\bar{\mathbf{R}}$ satisfies\cite{ref25}
	\begin{equation}
		\label{func38}
		{\bar{\lambda }^{i}_{1}}/\sigma _{n}^{2}\sim \mathcal{N}\left( \left( 1+{{\gamma }_{i}} \right)\left( 1+\frac{M-2}{2N{{\gamma }_{i}}} \right),\frac{{{\left( 1+{{\gamma }_{i}} \right)}^{2}}}{2N} \right).
	\end{equation}
	where ${{\gamma }_{0}}={{\left\| {\bar{\mathbf{h}}_{1}} \right\|}^{2}} \gamma$ and $ {{\gamma }_{1}}={{\left\| {\bar{\mathbf{h}}_{2}} \right\|}^{2}}K \gamma$. 
	
	According to \cite{related_semi3}, the OVL of ${\bar{\lambda}^{0}_{1}}/\sigma^{2}_{n}$ and ${\bar{\lambda}^{1}_{1}} / \sigma_{n}^{2}$ can be expressed as
	\begin{align}\label{OVL}
		& \text{OVL} \nonumber \\ 
		&=   \mathcal{Q}\left( \frac{\sqrt{2N}\gamma \! - \! \frac{M-2}{\sqrt{2N}\gamma {{\left\| {{\bar{\mathbf{{h}}}}_{1}} \right\|}^{2}} K {{\left\| {{\bar{\mathbf{{h}}}}_{2}} \right\|}^{2}}}}{\left( {{\left\| {{\bar{\mathbf{{h}}}}_{1}} \right\|}^{2}} + K {{\left\| {{\bar{\mathbf{{h}}}}_{2}} \right\|}^{2}} \right)\gamma +2}\left( \left| {{\left\| {{\bar{\mathbf{{h}}}}_{1}} \right\|}^{2}} \!- \! K {{\left\| {{\bar{\mathbf{{h}}}}_{2}} \right\|}^{2}} \right| \right) \right).
	\end{align}
	When $\Delta \gamma \rightarrow 0$ , $N \gg M$, and $\gamma {{\left\| {{\bar{\mathbf{{h}}}}_{2}} \right\|}^{2}} \ge 1$, \eqref{OVL} is approximated as
	\begin{align} \label{eq:eq61}
		\text{OVL} & = \mathcal{Q}\left( \frac{\sqrt{2N}\gamma - \frac{M-2}{\sqrt{2N}\gamma {{\left\| {{\bar{\mathbf{{h}}}}_{1}} \right\|}^{2}}{{\left\| {{\bar{\mathbf{{h}}}}_{2}} \right\|}^{2}}}}{\left( {1} + K \frac{{\left\| {{\bar{\mathbf{{h}}}}_{2}} \right\|}^{2}}{{\left\| {{\bar{\mathbf{{h}}}}_{1}} \right\|}^{2}} \right)\gamma + \frac{2}{{\left\| {{\bar{\mathbf{{h}}}}_{1}} \right\|}^{2}}}\left( \left| {1} \!- \! K \frac{{\left\| {{\bar{\mathbf{{h}}}}_{2}} \right\|}^{2}}{{\left\| {{\bar{\mathbf{{h}}}}_{1}} \right\|}^{2}} \right| \right) \right) \nonumber\\
		& \approx \mathcal{Q}\left( \frac{\sqrt{2N}\gamma - \frac{M-2}{\sqrt{2N}\gamma {{\left\| {{\bar{\mathbf{{h}}}}_{1}} \right\|}^{2}}{{\left\| {{\bar{\mathbf{{h}}}}_{2}} \right\|}^{2}}}}{\gamma + \frac{2}{{\left\| {{\bar{\mathbf{{h}}}}_{1}} \right\|}^{2}}} \right) \nonumber\\
		& \approx \mathcal{Q}\left( \frac{\sqrt{2N}\gamma}{\gamma + \frac{2}{{\left\| {{\bar{\mathbf{{h}}}}_{1}} \right\|}^{2}}} \right).
	\end{align}
	
	Eq. \eqref{eq:eq61} shows that for a large value of $N$, the OVL approaches $0$. 
	Therefore, based on \eqref{eq:eq53} and \eqref{eq:eq54}, the distribution of the second largest eigenvalue $\lambda_{2}/\sigma^2_{n}$ can be approximated by the distribution of $\bar{\lambda}^1_1/\sigma^2_n$. 
	Lemma \ref{lem2} is proved.

	\bibliographystyle{IEEEtran}
	\bibliography{REFE}
	\end{document}